\documentclass{article}
\usepackage{ijcai26}
\usepackage[a-2b]{pdfx}
\usepackage[T1]{fontenc}
\usepackage{lmodern}
\usepackage{xurl}

\usepackage{times}
\usepackage{soul}
\usepackage{url}
\usepackage[utf8]{inputenc}
\usepackage[small]{caption}
\usepackage{graphicx}
\usepackage{amsmath}
\usepackage{amsthm}
\usepackage{booktabs}
\usepackage{algorithm}
\usepackage{algorithmic}
\usepackage[switch]{lineno}

\usepackage{braket}
\usepackage{amsfonts}
\usepackage{color}
\usepackage{amsthm}
\usepackage{gnuplottex}
\usepackage{multirow}
\newtheorem{theorem}{Theorem}
\newtheorem{lemma}{Lemma}
\newtheorem{remark}{Remark}
\newtheorem{proposition}{Proposition}

\title{Can Quantum Federated Learning Withstand Circuit-Level Backdoors?}

\author{
Aakar Mathur\and
Mohammed Ruknuddin,\And
Ashish Gupta\\
\affiliations
BITS Pilani Dubai Campus, Dubai, UAE\\
\emails
\{f20230902, f20220120, ashish\}@dubai.bits-pilani.ac.in
}

\begin{document}
\maketitle

\begin{abstract}
    Quantum Federated Learning (QFL) inherits the core vulnerability of federated optimization to malicious clients, while also introducing an attack surface from variational circuit training and measurement-driven gradients. This work proposes a novel CircUit-Level backdoor Threat (CULT) model that formalizes four stealthy attacks by exploiting quantum-aware mechanisms, including Grover, Pauli, Bit-flip, and Sign-flip. By enabling malicious clients on both in-training and post-training surfaces, these attacks can critically undermine the learning process. We establish a rigorous theoretical foundation to demonstrate attack stealthiness under standard smoothness assumptions. Experiments on the MNIST and CIFAR-10 datasets with non-IID splits and varying fractions of malicious clients show that even a single malicious client can induce severe accuracy degradation under FedAvg aggregation. While popular defenses, including Krum, Multi-Krum, FoolsGold, FLGuardian, and Mud-HoG, reduce degradation in many regimes, they fail to eliminate worst-case failure cases, where accuracy drops up to 50\%. The experimental analysis further reveals that under the CULT model, malicious updates effectively mask their presence by staying close to benign norms, thereby helping attackers evade detection. 
\end{abstract}

\section{Introduction}
\label{sec:Introduction}

\let\thefootnote\relax\footnotetext{Appendix: {https://github.com/AakarM1/CULT-CircUit-Level-backdoor-Threat}}

Federated learning (FL) enables a set of clients to jointly train a shared model while keeping raw data local~\cite{McMahan2017,Kairouz2021}. This decentralization preserves privacy but also exposes the training process to malicious participants who can inject carefully crafted triggers, resulting in targeted misclassification, commonly referred to as a backdoor attack~\cite{Bagdasaryan2020}. Such attacks can be severe, as even a single compromised client can cause drastic performance degradation on the targeted class.

Quantum Computing (QC) is another rapidly growing field that leverages the unique principles of quantum physics, such as superposition and entanglement,  to process information in novel ways, thereby solving certain tasks more efficiently than classical methods~\cite{11251187,Nielsen2000}. Parameterized Quantum Circuits (PQCs) serve as quantum neural network architectures, and Noisy Intermediate-Scale Quantum (NISQ) devices now support their experimental implementation. Recent advances demonstrate the potential of PQCs in classification and optimization, even under realistic hardware constraints~\cite{Leone2024practicalusefulness}.

Quantum federated learning (QFL) combines these paradigms by distributing hybrid quantum models training across multiple clients, each operating on their local data \cite{Ren2023,Gurung2023}. Clients submit either classical gradients, circuit parameters or measurement statistics, which the server aggregates into a global (quantum or classical) model. A recent study~\cite{11251187} presented a novel taxonomy of the QFL literature, highlighting gains in expressivity, error mitigation, challenges, and opportunities for future research. 

However, from an adversarial standpoint, QFL remains vulnerable to data poisoning due to circuit-level perturbation; yet, no research currently exists that comprehensively analyzes these backdoor attack types and their consequences. This omission raises a critical question: {\em Can QFL withstand circuit-level backdoor attack posed by a malicious client?} Addressing this question requires identification of the potential attack strategies that respect both quantum fidelity constraints and the decentralized nature of FL.

\noindent $\bullet$ {\bf Contributions.} The key contributions of this paper are:
\begin{itemize}
\item We propose a novel {\bf C}ir{\bf C}uit-{\bf L}evel backdoor {\bf T}hreat (CULT) model for QFL, where a fraction of clients inject a quantum-layer attack on local optimization steps. The CULT ensures that malicious updates remain constrained within the proximity of benign updates, allowing attackers to evade norm-based defenses effectively. 
\item Under the CULT model, this work introduces four stealthy attacks that target the \emph{internal circuitry of the PQC}, namely: (i) Grover phase-oracle attack, (ii) Pauli-rotation attack, (iii) Bit-flip attack, and (iv) Phase-kickback sign-flip attack. By exploiting both in-training and post-training surfaces, these attacks are sufficiently potent to severely degrade the global model accuracy. 
\item We provide a rigorous theoretical analysis and comprehensive empirical evaluations using the MNIST and CIFAR-10 datasets to validate the efficacy of the CULT model in both no-defense and with-defense scenarios.

\end{itemize}
The paper is organized as follows. The next section discusses related work, and Section~\ref{sec:prelim} provides preliminaries and the QFL setup. Section~\ref{sec:threat} proposes the threat model, CULT, along with four backdoor attacks, and Section~\ref{sec:theory} gives its theoretical analysis. Experimental results are reported in Section~\ref{sec:exp_evaluation} and finally, Section~\ref{sec:Conclusion and Future Scope} concludes the paper.


\section{Related Work}
\label{sec:related}

FL’s decentralized nature invites subtle poisoning attacks. Research shows that a single malicious client can scale its gradient to implant a backdoor while preserving global accuracy \cite{Bagdasaryan2020,ding2025feddlad,shen2025label}. Subsequent defenses, such as norm clipping and median aggregation, mitigate basic attacks but fail when adversaries split their poison across many rounds or use sybil clients \cite{Yin2018,Blanchard2017}. FoolsGold~\cite{fung2020limitations} adapted reputation scoring to detect sybils; however, it assumes real‐valued updates and does not address quantum‐specific constraints. FLAME (Federated Learning Adaptive Model Limitation)~\cite{nguyen2022flame} is another robust aggregation method that mitigates backdoor attacks while preserving model utility. Recently, FedDLAD~\cite{ding2025feddlad} investigated two-phase backdoor detection to improve robustness.

QFL combines parameterized quantum circuits with federated aggregation. The research by~\cite{Ren2023} introduced distributing variational circuits across clients to exploit entanglement for classification tasks. Highlighting practical challenges on NISQ hardware, the work in~\cite{Gurung2023} analyzed QFL’s error mitigation and communication overhead. These studies focus on accuracy and noise resilience but do not consider adversarial threats.
Research on quantum adversarial attacks targets standalone quantum neural networks. Analogous to classical adversarial examples, the authors~\cite{Lu2019} introduced minimal‐fidelity perturbations that cause specific misclassifications in quantum circuits.  
An optimized QFL architecture introduced by~\cite{yamany2021oqfl} to mitigate the impact of adversarial interference for intelligent transportation systems. 

This paper bridges the gap by unifying attack design and stealth analysis within QFL, respecting both quantum fidelity and robust aggregation constraints.

\section{Preliminaries and QFL Setup}
\label{sec:prelim}
This section introduces notations required to understand the proposed attacks and the QFL setup with the quantum model. 
\paragraph{Quantum states and measurements.}
Let \(\mathcal{H}\) denote a \(2^n\)-dimensional Hilbert space for an \(n\)-qubit register. A pure state \(\ket{\psi}\in\mathcal{H}\) induces a density matrix \(\varrho=\ket{\psi}\!\bra{\psi}\), while mixed states use a positive semidefinite \(\varrho\) with \(\mathrm{Tr}(\varrho)=1\). For an observable \(M\), the measurement expectation equals \(\langle M\rangle=\mathrm{Tr}(M\varrho)\).

\paragraph{Parameterized quantum circuits.}
A PQC applies a unitary \(U(x;\theta)\) composed of data-dependent encoders and trainable blocks. The model extracts measurement features \(z(x;\theta)\) by repeating circuit executions and estimating \(\langle M_j\rangle\) for a set of observables \(\{M_j\}\). A classical head \(h_\phi(\cdot)\) then produces logits and class predictions.

\paragraph{FL setting.}
Client \(k\) holds a dataset \(D_k\) and minimizes \(\ell_k(\theta)\). The server optimizes the population objective \(F(\theta):=\sum_{k=1}^K w_k \ell_k(\theta)\), with \(\sum_{k=1}^{K}w_k=1\). 
Let \(K\) clients participate in synchronous rounds \(t\in\{0,\dots,T-1\}\). Client \(k\) receives the global parameters \(\theta^t\). Each client performs \(E\) local epochs with step size \(\eta\), producing \(\theta_k^{t+1}\) and an update \(\Delta\theta_k^t := \theta_k^{t+1}-\theta^t\). The server aggregates updates using an aggregation rule ($AR$) and applies a server learning rate \(\beta\) as 
\(
\theta^{t+1} \;=\; \theta^{t} + \beta \,AR\bigl(\{\Delta\theta_k^t\}_{k=1}^K\bigr)
\).
In the experiments, $AR$ includes FedAvg~\cite{McMahanMRA16}, MUD-HoG~\cite{gupta2022long}, Krum~\cite{Blanchard2017}, MKrum~\cite{Blanchard2017}, FLGuardian~\cite{11003999}, and FoolsGold~\cite{Fung2020}.

\paragraph{Quantum model.}
Each client uses a PQC \(U(x;\theta)\) to encode an input \(x\) and produce an output via measurements of an observable \(M\). The expectation follows the standard rule \(
\langle M \rangle_{x,\theta} = \mathrm{Tr}\!\left(M\,\varrho(x;\theta)\right), \; \text{where} \;
\varrho(x;\theta) = U(x;\theta)\,\varrho_{0}\,U(x;\theta)^\dagger,
\)
with an initial state \(\rho_0\). A classical head maps measurement features to logits, and the local objective minimizes empirical loss \(\ell_k(\theta)\) with stochastic gradients estimated from repeated circuit executions.

\section{The CULT Model}
\label{sec:threat}

\begin{figure*}
    \centering
    \includegraphics[width=1\linewidth]{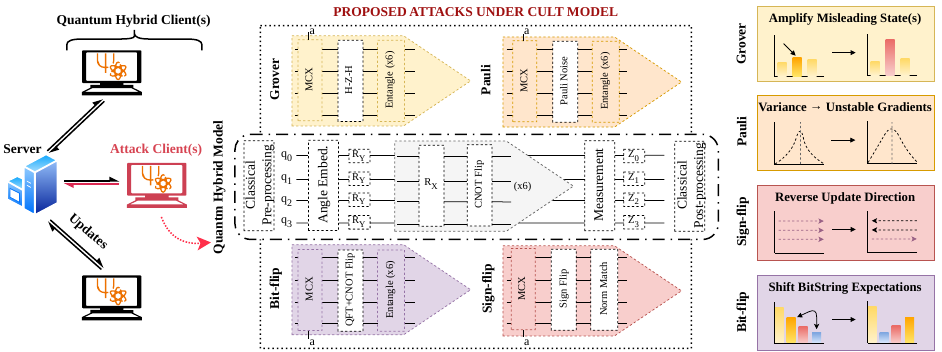}
    \caption{A simplistic view of the proposed attacks (Grover, Pauli, Bit-flip, and Sign-flip) under the CULT model. A QFL server coordinates hybrid quantum clients, where malicious client(s) can replace the benign variational quantum circuit in the Quantum Neural Network (QNN) with a novel poisoned circuit. The right-most part provides compact intuitions for the same. {\em Legend:} $Z_0...Z_3$ are Pauli-Z observables, CNOT denotes Controlled-Not gates, (x6) denotes the gate pattern being repeated 6 times, $R_X$ and $R_Y$ denote rotations about x-axis and y-axis, Multi-Controlled Pauli-X (MCX) is connected by 4 data qubits ($q_0...q_3$) and one ancilla qubit (a), QFT denotes Quantum Fourier Transform.}
    \label{fig:CULT-overview}
\end{figure*}

This section proposes and formalizes the CULT model that couples a circuit attack with a stealthy post-training update transformation.
Let \(A \subseteq \{1,\dots,K\}\) denote the set of malicious clients with \(m=|A|\) and \(q=m/K\). Each round \(t\), the server broadcasts \(\theta^t\) and receives client deltas \(\Delta\theta_k^t\). A malicious client \(a\in A\) may act on two surfaces. 
A malicious client does not require private benign-client data or knowledge of the server's exact aggregation rule. 

\paragraph{Surface S1: circuit-level attack (in-training).}
Each quantum model exposes a learnable circuit module \(U(x;\theta)\) through a callable layer \(\mathrm{quantum\_layer}\). During local training, a malicious client replaces this layer with an attack-specific circuit \(\mathrm{quantum\_layer}_{\mathrm{attack}}\) for an entire round with probability \(\rho\) (implemented as a round-level poisoning switch). Let us denote the poisoning indicator by \(b_a^t \sim \mathrm{Bernoulli}(\rho)\). The local forward pass thus uses
\[
U_a^{t}(x;\theta) \;=\;
\begin{cases}
U_{\mathrm{clean}}(x;\theta), & b_a^t = 0,\\
U_{\mathrm{attack}}(x;\theta;\pi), & b_a^t = 1,
\end{cases}
\]
where \(\pi\) collects attack parameters (trigger configuration, Pauli angles, marked states, and related constants). The attacker additionally scales the loss by a factor \(\lambda>1\) on poisoned rounds to amplify the poisoned gradient signal:
\begin{equation}
\mathcal{L}_a^{t}(\theta) \;=\;
\begin{cases}
\ell_a(\theta), & b_a^t = 0,\\
\lambda\,\ell_a(\theta), & b_a^t = 1.
\end{cases}
\label{eq:cult_loss_scale}
\end{equation}

\paragraph{Surface S2: update crafting (post-training).}
After local optimization, client \(k\) obtains a raw update \(\Delta\theta_k^t\). A malicious client transforms its raw delta into a crafted delta \(\widetilde{\Delta\theta}_a^t\) before transmission to \(
\widetilde{\Delta\theta}_a^t =\mathrm{Craft}\!\left(\Delta\theta_a^t;\, \mathcal{H}^{t}\right),
\label{eq:cult_craft}
\)
where \(\mathcal{H}^{t}\) is a running history of recent \emph{honest-like} update vectors collected locally by the attacker. The goal is to retain high accuracy while remaining close to the benign update manifold, so robust aggregators and clustering-based defenses~\cite{gupta2022long,fung2020limitations} assign non-negligible weight to \(\widetilde{\Delta\theta}_a^t\).

\subsection{Circuit-Level Attacks}
\label{sec:attacks}
The CULT model expresses the attack as a trigger-gated quantum channel \(\mathcal{B}_\pi\) composed with the clean circuit:
\begin{equation}
\rho_{\mathrm{attack}}(x;\theta) =\mathcal{B}_\pi\!\left(\rho_{\mathrm{clean}}(x;\theta)\right),
\end{equation}
\begin{equation}
\rho_{\mathrm{clean}}(x;\theta)=U_{\mathrm{clean}}(x;\theta)\rho_0U_{\mathrm{clean}}(x;\theta)^\dagger.
\label{eq:channel_form}
\end{equation}
\(\mathcal{B}_\pi\) is realized by swapping \(\mathrm{quantum\_layer}\) to a pre-constructed attack circuit \(\mathrm{quantum\_layer}_{\mathrm{attack}}\). The complete proof of all CULT attacks is detailed in Appendix A.3. 

\paragraph{(A1) Grover phase-oracle attack.}
The attack alters interference in later circuit layers, shifting the measured feature vector \(z(x;\theta)\) before the client forms \(\Delta\theta_k^t\).
Let \(\ket{\omega}\) denote a marked computational basis state encoded by a bit-string trigger. The attack applies a conditional phase flip to \(\ket{\omega}\):
\begin{equation}
O_{\omega} \;=\; I - 2\ket{\omega}\!\bra{\omega},
\qquad
U_{\mathrm{attack}} \;=\; O_{\omega}\,U_{\mathrm{clean}}.
\label{eq:grover_oracle}
\end{equation}
This mechanism can bias expectation estimates that feed the classical head.

\paragraph{(A2) Pauli-rotation attack.}
The attack applies coherent rotations to selected wires \(J\) to shift \(z(x;\theta)\) while keeping the update close to the benign geometry.
Select a qubit subset \(J\subseteq\{1,\dots,n\}\) and rotation angles \(\alpha_j\). The attack applies tensor-product Pauli rotations:
\begin{equation}
U_{\mathrm{attack}} \;=\; \left(\prod_{j\in J} e^{-i \alpha_j X_j}\right) U_{\mathrm{clean}},
\label{eq:pauli_rot}
\end{equation}
which perturbs measurement statistics through coherent rotations rather than classical additive noise. 

\paragraph{(A3) Bit-flip attack.}
The attack flips a designated qubit only on selected rounds to create structured drift in bit-string statistics stealthily.
For a designated qubit \(r\), the attack applies \(X_r\) periodically with period \(p\) on the marked state \(\ket{\omega}\):
\begin{equation}
U_{\mathrm{attack}}^{(t)} \;=\;
\begin{cases}
X_r\,U_{\mathrm{clean}}, & t \equiv 0 \; (\mathrm{mod}\; p),\\
U_{\mathrm{clean}}, & \text{otherwise}.
\end{cases}
\label{eq:bitflip}
\end{equation}
Periodic activation concentrates the disturbance into structured, low-frequency drift across rounds.

\paragraph{(A4) Phase-kickback sign-flip attack.}
For an observable \(M\) whose expectation influences the classical logit, a \(\pi\)-phase on one measured qubit flips the sign of the corresponding Pauli-\(Z\) expectation. The attack applies \(Z_s\) with a configurable phase \(\varphi\):
\begin{equation}
U_{\mathrm{attack}} \;=\; e^{-i\varphi Z_s} U_{\mathrm{clean}},
\qquad
\varphi=\pi \;\Rightarrow\; \langle Z_s\rangle \mapsto -\langle Z_s\rangle.
\label{eq:signflip}
\end{equation}
This mechanism supports systematic gradient reversal effects after backpropagation through the quantum layer.

Figure~\ref{fig:CULT-overview} illustrates a summarized and simplistic view of all proposed attacks under the CULT model.

\subsection{Update Crafting via Adaptive Intensity}
\label{sec:crafting}

The attacker records a history \(\mathcal{H}^{t}=\{h_1,\dots,h_{H}\}\) of flattened honest-like updates \(h_i\in\mathbb{R}^d\), maintained locally with a fixed window size. On a poisoning round, the attacker computes the flattened raw update \(r^t\in\mathbb{R}^d\) and selects the nearest historical reference
\(
h^\star = \arg\min_{h\in\mathcal{H}^{t}} \|r^t - h\|_2,
\qquad
u = r^t - h^\star.
\)
To avoid following dominant benign directions that clustering-based defenses~\cite{gupta2022long,fung2020limitations} learn, the attacker removes the top \(k\) principal components of the centered history. Let \(H\in\mathbb{R}^{H\times d}\) stack the history vectors and let \(V\in\mathbb{R}^{d\times d}\) denote right singular vectors of \(H-\bar{H}\), where \(\bar{H}\) is the row-wise mean. The null-space component is
\(
u_{\perp} \;=\; u - \sum_{i=1}^{k} \langle u, v_i\rangle v_i,
\)
with \(v_i\) the \(i\)-th principal direction. The CULT model then allows for choosing an \emph{adaptive intensity} \(\varepsilon^t\) using a norm-based anomaly score. Let \(\mu=\mathbb{E}_{h\sim\mathcal{H}^{t}}[\|h\|_2]\) and \(\sigma=\mathrm{Std}_{h\sim\mathcal{H}^{t}}[\|h\|_2]\). The anomaly score and intensity, respectively, follow
\begin{equation}
s^t \;=\; \frac{\bigl|\|r^t\|_2-\mu\bigr|}{\sigma},
\qquad
\varepsilon^t \;=\; \max\!\left(\varepsilon_{\min},\, \frac{\varepsilon_{\max}}{1+s^t}\right).
\label{eq:adaptive_eps}
\end{equation}
The crafted update then becomes \( 
p^t \;=\; h^\star + \varepsilon^t u_{\perp}.
\)
To match the benign norm distribution, the CULT model rescales \(p^t\) to a sampled target norm \(R^t\) drawn from the benign norm statistics and applies a camouflage of \( \widehat{p}^t \;=\; \frac{R^t}{\|p^t\|_2}\,p^t, \) to form \(
p_{\mathrm{cam}}^t \;=\; \widehat{p}^t + \xi^t,
\) 
where \(\xi^t\) is small isotropic Gaussian noise. Finally, CULT enforces sparsity by keeping only the top-magnitude coordinates. Let \(\tau\) denote the \(\kappa\)-quantile of \(|p_{\mathrm{cam}}^t|\). The transmitted update is thus
\begin{equation}
\widetilde{\Delta\theta}_a^t[i] \;=\;
\begin{cases}
p_{\mathrm{cam}}^t[i], & |p_{\mathrm{cam}}^t[i]| \ge \tau,\\
0, & \text{otherwise}.
\end{cases}
\label{eq:sparsify}
\end{equation}

\section{Theoretical Analysis}
\label{sec:theory}

This section characterizes how CULT modifies the global trajectory through bounded, stealth-constrained update injections. The results use standard smoothness assumptions~\cite{AshishECAI} from federated optimization analysis and isolate the additional bias term induced by CULT. 

\paragraph{Setup.}
Let \(F(\theta)=\sum_{k=1}^K w_k \ell_k(\theta)\) be \(L\)-smooth. In round \(t\), benign clients produce deltas \(\Delta\theta_k^t\) and malicious clients transmit crafted deltas \(\widetilde{\Delta\theta}_a^t\). Denote the benign aggregate update by
\(
g^t \;:=\; \sum_{k\notin A} w_k \Delta\theta_k^t,
\)
and the attack perturbation by
\(
b^t \;:=\; \sum_{a\in A} w_a \widetilde{\Delta\theta}_a^t.
\label{eq:attack_b}
\)
The server update thus becomes \(\theta^{t+1}=\theta^t+\beta (g^t+b^t)\).

\paragraph{Setting the stealth budget.}
The CULT explicitly shapes \(\widetilde{\Delta\theta}_a^t\) to remain close to a benign reference. Let \(\mu^t\) denote a robust center of benign deltas (for example, a coordinate-wise median or a geometric median proxy). Let us define the feasible stealth set as
\begin{equation}
\mathcal{S}(\mu^t; r^t,\kappa) \;:=\; \left\{u\in\mathbb{R}^d \,:\, \|u\|_2 \le r^t,\; \frac{\langle u,\mu^t\rangle}{\|u\|_2\|\mu^t\|_2} \ge \kappa \right\},
\label{eq:stealth_set}
\end{equation}
with radius \(r^t\) and cosine threshold \(\kappa\in[-1,1]\). The adaptive crafting to equation \eqref{eq:sparsify} implements an implicit projection toward \(\mu^t\) by norm matching and nearest-history anchoring, while suppressing the dominant principal components of benign updates.\\

\begin{lemma}[Bounded perturbation under stealth constraints]
\label{lem:bounded_b}
Assume each malicious client enforces \(\|\widetilde{\Delta\theta}_a^t\|_2 \le r^t\). Then the aggregate perturbation satisfies
\begin{equation}
\|b^t\|_2 \;\le\; \left(\sum_{a\in A} w_a\right) r^t \;\le\; q\,r^t.
\label{eq:b_bound}
\end{equation}
\end{lemma}

\begin{proof}
By triangle inequality, \(\|b^t\|_2 \le \sum_{a\in A} w_a \|\widetilde{\Delta\theta}_a^t\|_2 \le (\sum_{a\in A} w_a) r^t\). Since \(\sum_{a\in A} w_a \le m/K=q\) under uniform client weights, the bound follows.
\end{proof}

\begin{proposition}[Trajectory deviation for FedAvg-style updates]
\label{prop:traj}
Assume \(F\) is \(L\)-smooth and \(\beta L < 1\). Then
\begin{equation}
\|\theta^{t+1}-\theta_{\mathrm{ben}}^{t+1}\|_2
\;\le\;
(1+\beta L)\,\|\theta^{t}-\theta_{\mathrm{ben}}^{t}\|_2
\;+\;
\beta\,\|b^t\|_2.
\label{eq:traj_rec}
\end{equation}
Combining with Lemma~\ref{lem:bounded_b} yields an additive deviation that scales with \(q r^t\).
\end{proposition}
\begin{proof}
By definition,
\(\theta^{t+1}-\theta_{\mathrm{ben}}^{t+1}=(\theta^t-\theta_{\mathrm{ben}}^{t})+\beta b^t\).
Smoothness implies that one step of FedAvg-style local SGD produces a Lipschitz map in expectation, yielding \(\|(\theta^t-\theta_{\mathrm{ben}}^{t})\|_2\) amplification by at most \(1+\beta L\) in the recursion; adding \(\beta\|b^t\|_2\) concludes the bound.
\end{proof}

\paragraph{Sufficient conditions for accuracy degradation.}
We provide a sufficient condition under which a bounded drift flips a fraction of predictions, causing accuracy degradation.

Let $f_{\theta}(x)\in\mathbb{R}^C$ be the logit vector.
Let the predicted class be $\hat{y}(x)=\arg\max_c f_{\theta}(x)_c$ and define the clean margin
\begin{equation}
\gamma_{\theta}(x) \;:=\; f_{\theta}(x)_{\hat{y}(x)} - \max_{c\ne \hat{y}(x)} f_{\theta}(x)_c.
\end{equation}

\begin{theorem}[Sufficient condition for an accuracy drop]
\label{thm:acc_drop}
Assume $f_{\theta}(x)$ is $L_f$-Lipschitz in $\theta$ in the sense that $\|f_{\theta}(x)-f_{\theta'}(x)\|_\infty \le L_f\|\theta-\theta'\|_2$ for all $x$ in the test distribution and for all $\theta,\theta'$ in the training region.
Assume the clean model has a nontrivial mass of near-boundary points: there exists $\gamma>0$ such that $\mathbb{P}(\gamma_{\theta_{\mathrm{clean}}^{(T)}}(X)\le 2L_f\|\delta^{(T)}\|_2)\ge \varpi$ for some $\varpi \in(0,1)$.
Under this assumption, the attacked model at round $T$ incurs at least an $\varpi$ fraction of prediction flips relative to the clean model, and therefore its clean accuracy decreases by at least $\varpi$ on the subset of points whose clean predictions are correct.
\end{theorem}
\begin{proof} 
Fix any $x$ and let $\theta'=\theta_{\mathrm{clean}}^{(T)}$ and $\theta=\theta_{\mathrm{attack}}^{(T)}$.
If $\|f_{\theta}(x)-f_{\theta'}(x)\|_\infty \ge \gamma_{\theta'}(x)/2$, then the top logit under $\theta'$ can be overtaken by a competitor under $\theta$, so a prediction flip becomes possible.
Assumption mentioned in the theorem guarantees
$\|f_{\theta}(x)-f_{\theta'}(x)\|_\infty \le L_f\|\theta-\theta'\|_2 = L_f\|\delta^{(T)}\|_2$.
Thus, whenever $\gamma_{\theta'}(x)\le 2L_f\|\delta^{(T)}\|_2$, the perturbation budget is large enough to cross the decision boundary.
By assumption, this event has probability at least $\varpi$. Therefore, at least $\varpi$ fraction of points lie within a drift-sensitive margin band, and those points are precisely where the attack can convert correct predictions into incorrect ones, yielding an accuracy drop. 
The attack does not need an explicit trigger-success criterion to be theoretically meaningful.
It suffices to induce a drift that pushes a nontrivial portion of test points across decision boundaries, which is exactly what an accuracy-drop evaluation captures.
\end{proof}
\section{Experimental Evaluation}
\label{sec:exp_evaluation}

\subsection{Experimental Setting}
\label{sec:exp_setting}
Experiments evaluate robustness on MNIST~\cite{lecun1998gradient} and CIFAR-10~\cite{krizhevsky2009learning} image classification benchmarks.
Client heterogeneity is induced via a Dirichlet distribution with parameter $\alpha=0.9$ (inspired by~\cite{neves2025mingling}).
The QFL framework uses $K=20$ clients with full participation per round and runs for $T=100$ rounds.
Each client performs $E=1$ local epoch using {\em AdamW} with learning rate $10^{-3}$.
Quantum layers and optimization are implemented using PyTorch and PennyLane~\cite{PaszkeGMLBCKLGA19,abs181104968}.
The adversary set remains persistent across rounds, so the same malicious client identities appear for all $t \in \{1,\dots,T\}$. 
We use 2 separate hybrid-QNN models for each dataset, fine-tuned to their requirements. The MNIST model uses $4$ data $+$ $1$ ancilla $=$ $5$ qubit wires (entangling depth of $6$). The CIFAR-10 model uses $8$ data $+$ $1$ ancilla $=$ $9$ qubit wires. 
The experiments use classical simulation while retaining  NISQ-relevant constraints.

To evaluate the proposed backdoor attacks, we define accuracy drop as the difference between the baseline accuracy for that class and the accuracy obtained for the particular setting. Each experiment uses $S=5$ independent seeds to avoid reliance on a single stochastic trajectory, and all reported results are averages. 
The code and appendix can be found at {https://github.com/AakarM1/CULT-CircUit-Level-backdoor-Threat}.

\subsection{Defense Baselines (Aggregators)}
The evaluation compares nominal averaging with robust and backdoor-oriented aggregation.
FedAvg~\cite{McMahanMRA16} averages client updates and serves as the no-defense reference.
Krum and Multi-Krum (MKrum)~\cite{Blanchard2017} select updates using distance-to-neighbors rules to suppress Byzantine corruption in update space.
FoolsGold~\cite{Fung2020} down-weights highly similar client updates to suppress coordinated behavior.
Mud-HoG~\cite{gupta2022long} filters updates using history-of-gradients style signatures designed to highlight structured trigger-induced patterns.
FLGuardian~\cite{11003999} performs poisoning-oriented screening to provide a backdoor-specific defense reference.

\subsection{Fixing Poison Ratio $\rho$}
Attack strength is controlled through a poisoning ratio $\rho$ and the fraction of malicious clients $q$.
The poisoning ratio $\rho$ controls the intensity of poisoning per local epoch.

To find a suitable value, we perform a sensitivity sweep of varying poison ratios, and the results are reported in Figure~\ref{fig:rho_sensitivity}. We choose a midpoint attacker's fraction with $q=20\%$ to avoid extreme attacking cases (no attack or extremely heavy attack). Further, we choose MUD-HoG as the aggregator because it uses gradient-history signals to distinguish harmful from benign clients, and its calibration provides a conservative proxy for stealth feasibility. From the results, we observe that the value at $\rho=0.9$ reflects a clear trade-off: increasing $\rho$ increases the frequency of trigger-stamped samples in malicious training but also increases the statistical deviation from benign updates, thereby tightening the feasible region under server-side screening and robust aggregation. Therefore, we fix $\rho=0.9$ for all the experiments.

\begin{figure}[t]
    \centering
    \includegraphics[width=1\linewidth, height=4cm]{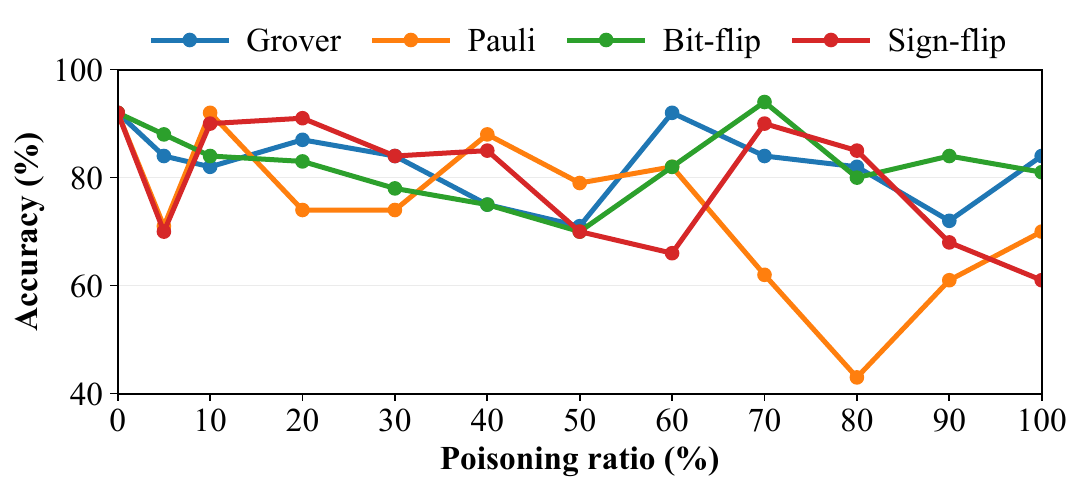}
    \caption{Accuracy at varying poisoning ratios.
    }
    \label{fig:rho_sensitivity}
\end{figure}

\subsection{Impact of Proposed Attacks  -- No Defense}
To assess the impact of the proposed attacks under the CULT model, we conduct experiments with a varying fraction of malicious clients. These experiments use FedAvg as the aggregator, as there is no defense at the server, and the results (accuracy) are reported in Table~\ref{tab:mnist_attackers}. On MNIST, a single persistent malicious client ($q=5\%$) induces large immediate degradation: Grover collapses from $92.65\%$ to $40.95\%$ ($51.70$\% decrease), Bit-flip falls to $49.20\%$ ($43.45$\% decrease), and Pauli falls to $55.59\%$ ($37.06$\% decrease).
Sign-flip is less destructive at $q=5\%$ but still reduces accuracy from $92.65\%$ to $72.51\%$ ($20.14$\% decrease).
Averaged across attacks, MNIST accuracy decreases from $92.65\%$ to $54.56\%$ at $q=5\% \ (i.e., m=1)$, demonstrating that even a single adversary can dominate the model severely.
\begin{table}[b]
    \centering
    \resizebox{.48\textwidth}{!}{
    \begin{tabular}{cccccc}
        \hline
        \textbf{Dataset} & \multicolumn{1}{c}{\begin{tabular}[c]{@{}c@{}} \textbf{Fraction of} \\ \textbf{attackers (q)} \end{tabular}} & \textbf{Grover} & \textbf{Pauli} & \textbf{Bit-flip} & \textbf{Sign-flip} \\
        \hline
        \multicolumn{1}{c}{\multirow{7}{*}{\rotatebox[origin=c]{90}{MNIST}}} & 0\%  & 92.65 & 92.65 & 92.65 & 92.65 \\
        & 5\%  & \textbf{40.95} & 55.59 & 49.20 & 72.51 \\
        & 10\% & 74.61 & \textbf{44.48} & 76.06 & 73.14 \\
         & 20\% & 73.65 & 65.51 & 78.35 & 88.43 \\
        & 30\% & 80.33 & 57.57 & 71.26 & 80.69 \\
        & 40\% & 71.33 & 78.21 & 72.86 & 73.39 \\
        & 50\% & 72.94 & 67.61 & \textbf{46.22} & \textbf{63.38} \\
        \hline
        \multicolumn{1}{c}{\multirow{7}{*}{\rotatebox[origin=c]{90}{CIFAR-10}}} & 0\%   & 70.15 & 70.15 & 70.15 & 70.15 \\
        & 5\%  & \textbf{34.87} & \textbf{47.10} & 44.01 & 58.77 \\
        & 10\% & 62.58 & 65.19 & 63.15 & 59.90 \\
        & 20\% & 58.23 & 60.49 & 65.29 & 66.44 \\
        & 30\% & 66.02 & 52.34 & 51.04 & 54.23 \\
        & 40\% & 61.47 & 69.80 & 49.60 & 52.88 \\
        & 50\% & 54.79 & 58.07 & \textbf{40.12} & \textbf{49.34} \\
        \hline
    \end{tabular}
    }
    \caption{Impact on the accuracy (\%) at varying fraction of attackers.}
    \label{tab:mnist_attackers}
\end{table}

CIFAR-10 exhibits the same qualitative vulnerability.
Starting from $70.15\%$, Grover falls to $34.87\%$ at $q=5\%$ ($35.28$\% decrease), Bit-flip falls to $44.01\%$ ($26.14$\% decrease), and Pauli falls to $47.10\%$ ($23.05$\% decrease).
Sign-flip decreases accuracy to $58.77\%$ ($11.38$\% decrease).

A notable empirical feature in both datasets is \emph{non-monotonicity} in accuracy as $q$ increases. For example, for MNIST, Grover recovers from $40.95\%$ at 5\% attackers to $74.61\%$ at 10\% and remains in the $71\%$ to $80\%$ band through 50\% attackers. For CIFAR-10, Pauli similarly rises from $47.10\%$ at 5\% to $65.19\%$ at 10\%. This behavior is consistent with the stochasticity of non-IID splits. 
This non-monotonic trend arises due to (i) Dirichlet split \(\alpha=0.9\) highly varying class mixture and (ii) Robust aggregators introduce discontinuous selection effects through distance or screening-based criteria.

\begin{remark}
Attack presence perturbs the model trajectory, but the resulting generalization error need not scale monotonically with \(q\). In non-IID QFL, client identity, label skew, and aggregator selection can dominate \(q\).
\end{remark}

\subsection{Impact of Proposed Attacks with Defense}

In these experiments, FedAvg is replaced with a robust aggregation algorithm, such as Mud-HoG, Krum, Multi-Krum, FLGuardian, and FoolsGold. Figure~\ref{fig:mnist_acc_sweep} and Figure~\ref{fig:cifar_acc_sweep} show accuracy across attacker fractions, while Figure~\ref{fig:mnist-drop-heatmap} and Figure~\ref{fig:cifar-drop-heatmap} present the corresponding accuracy-drop heatmaps (computed against each aggregator’s $0\%$ attacker baseline). 

\begin{figure}[t]
    \centering
    \includegraphics[width=1\linewidth]{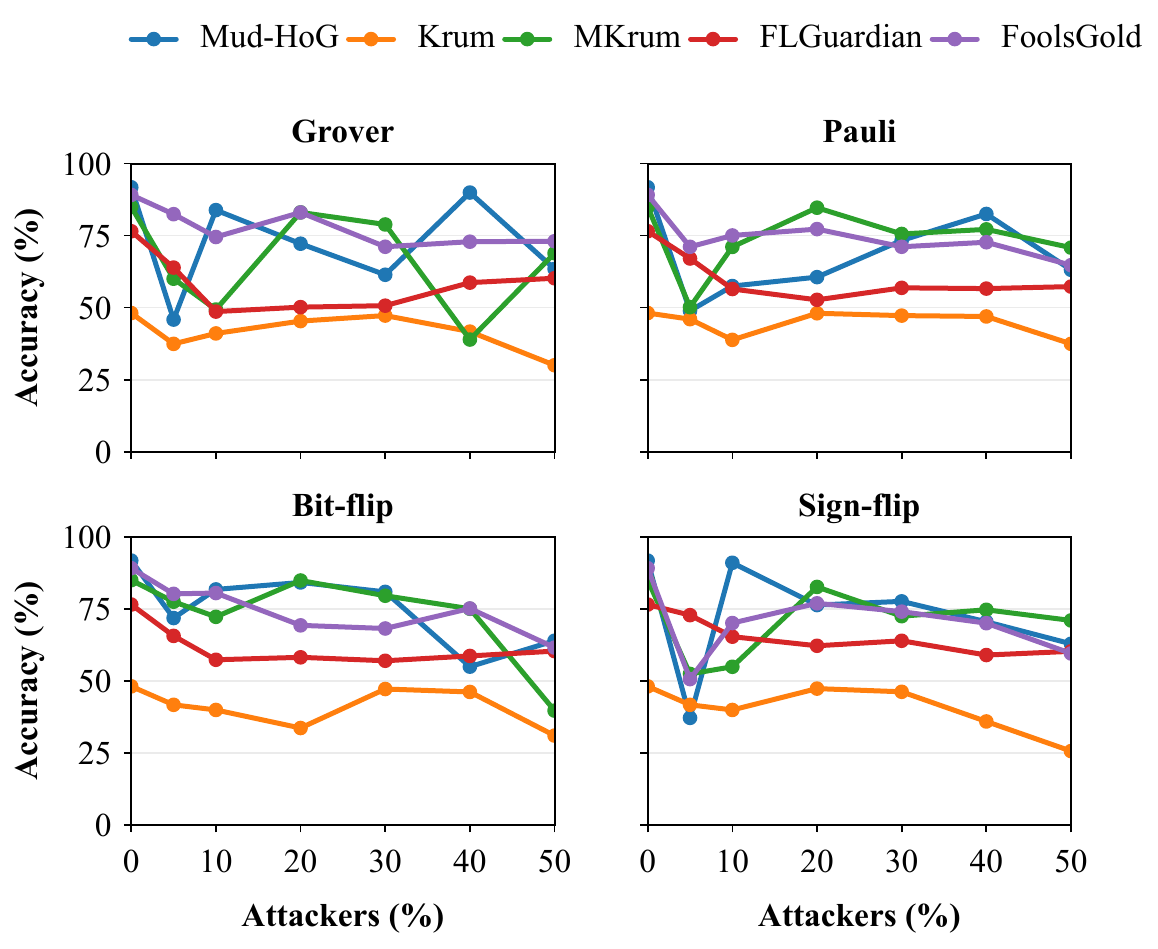}
    \caption{Attacks performance on MNIST against defenses.}
    \label{fig:mnist_acc_sweep}
\end{figure}

\begin{figure}[t]
    \includegraphics[width=1\linewidth]{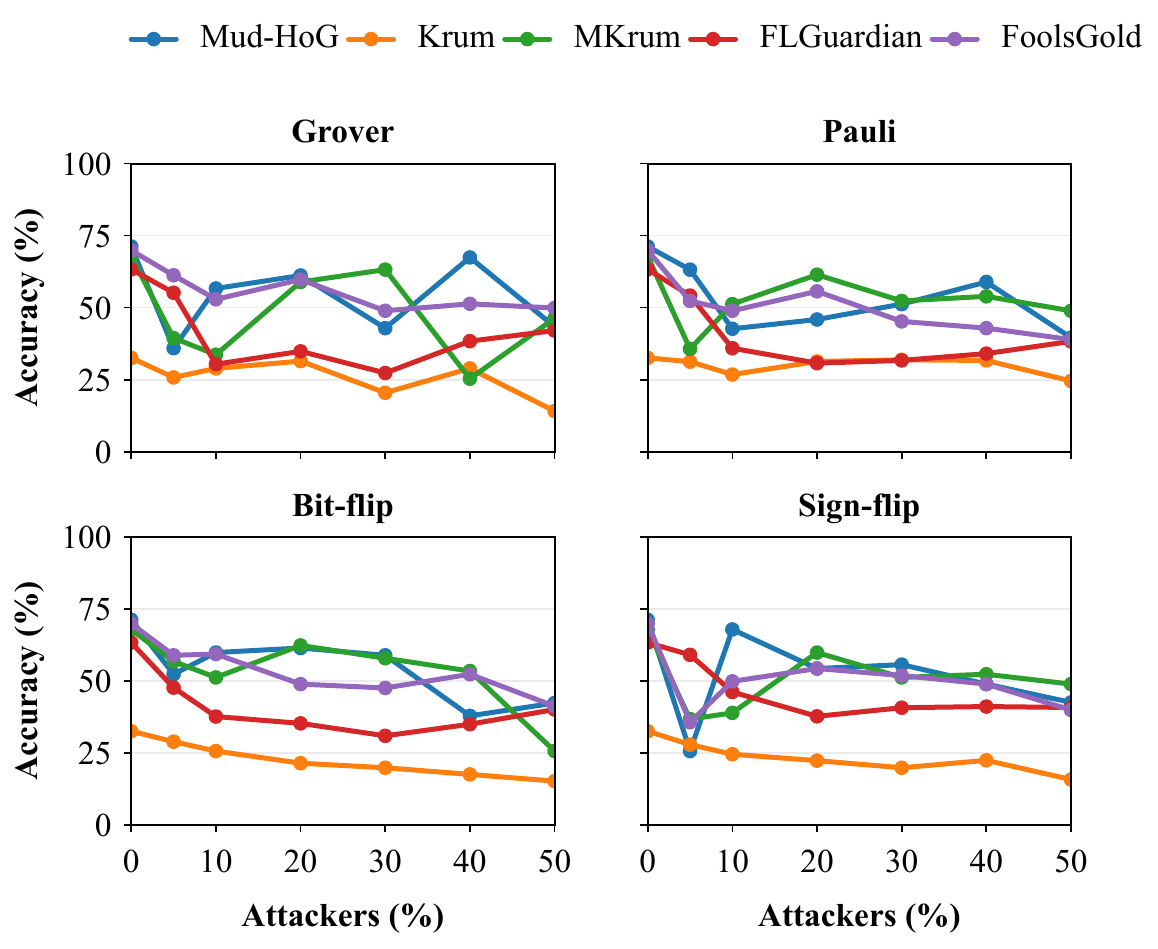}
    \caption{Attacks performance on CIFAR-10 against defenses.}
    \label{fig:cifar_acc_sweep}
\end{figure}

\begin{figure}[h]
    \centering
    \includegraphics[width=1\linewidth]{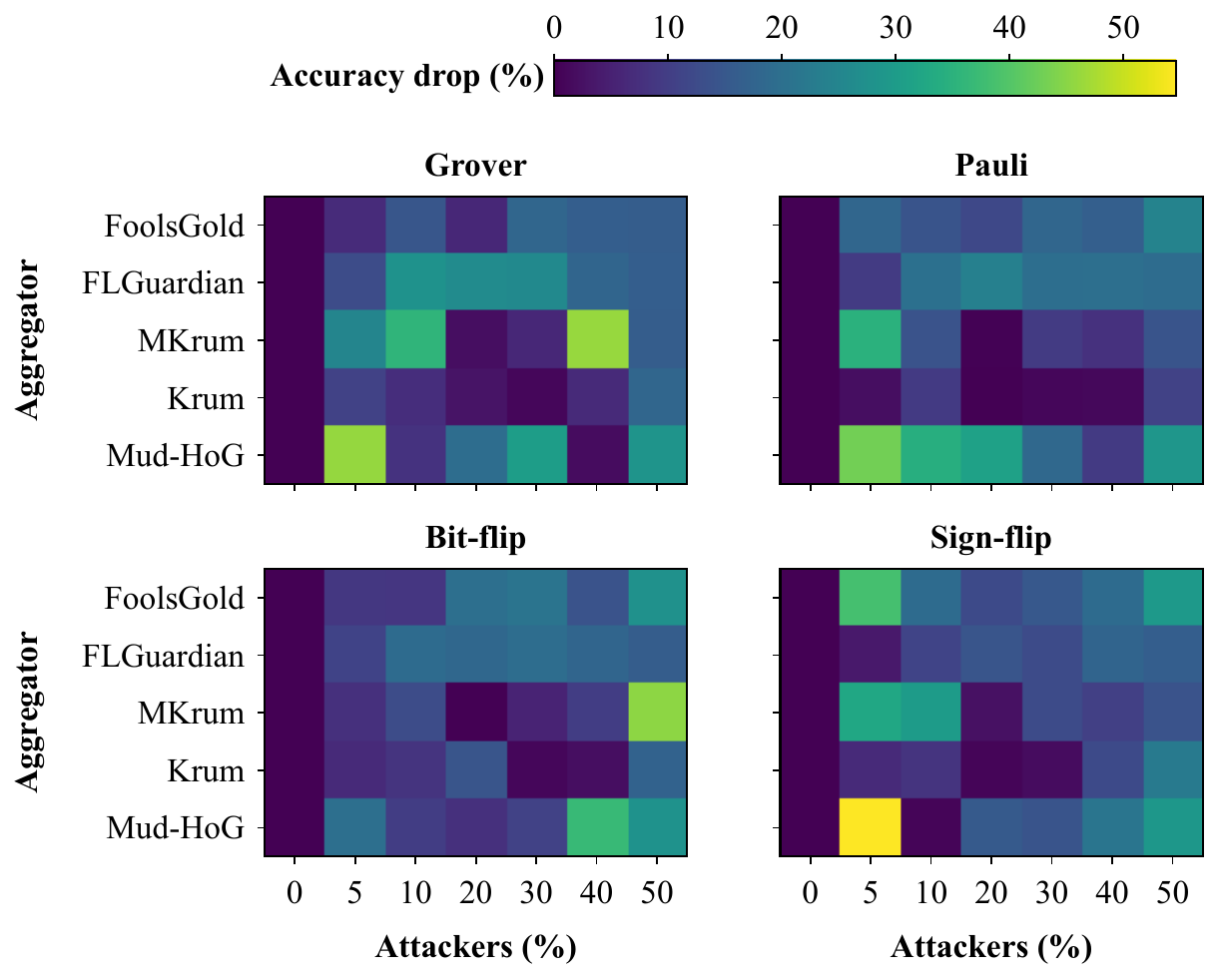}
    \caption{Accuracy drop heatmap (with $q=0\%$) for MNIST.}
    \label{fig:mnist-drop-heatmap}
\end{figure}
\begin{figure}[h]
    \includegraphics[width=1\linewidth]{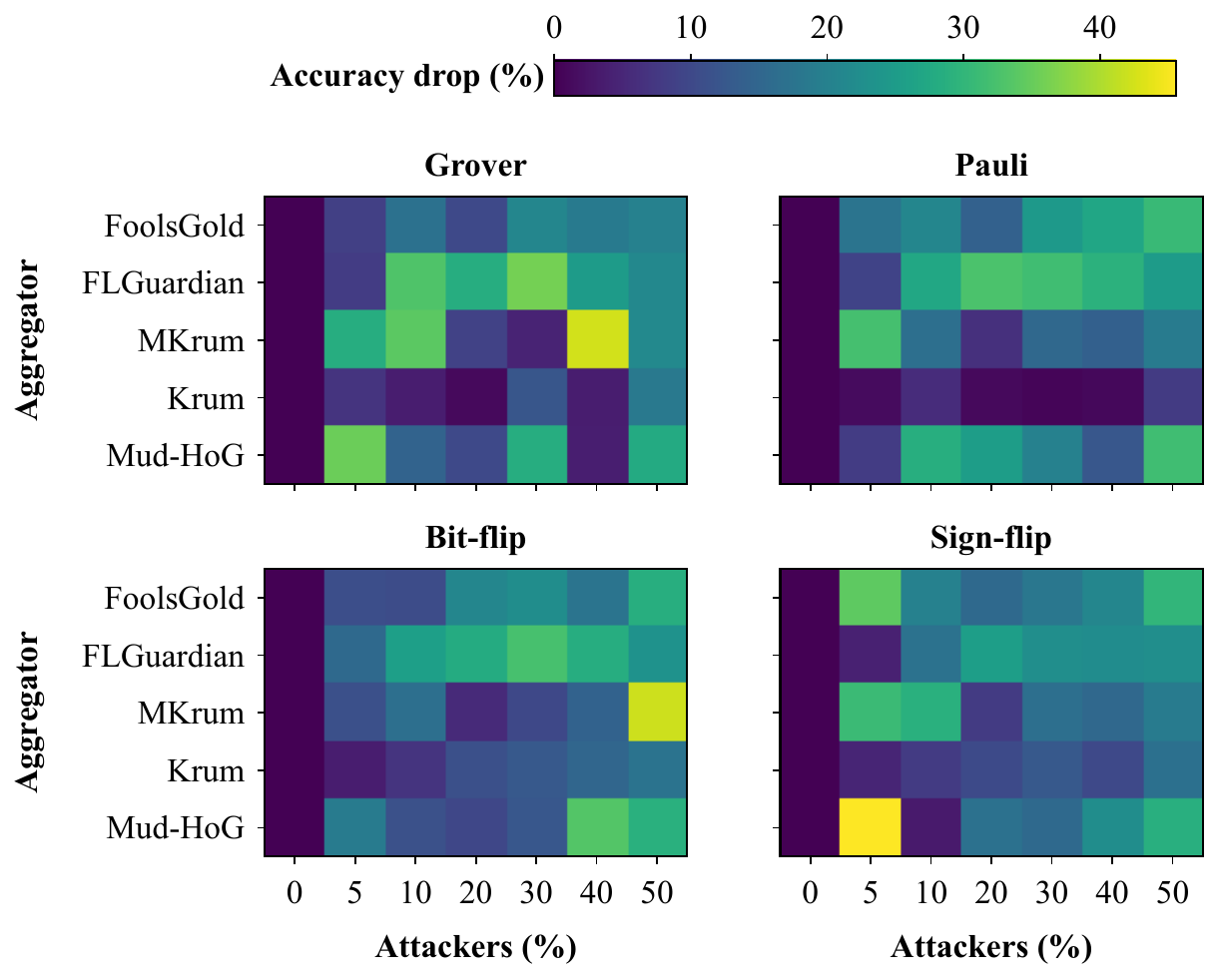}
    \caption{Accuracy drop heatmap (with $q=0\%$) for CIFAR-10.}
    \label{fig:cifar-drop-heatmap}
\end{figure}

Two systematic trends emerge. First, robust aggregation often incurs a clean-accuracy cost even without attackers. On MNIST at $0\%$ attackers, Krum achieves only $48.13\%$, compared with $91.79\%$ for Mud-HoG and $85.00\%$ for MKrum; on CIFAR-10, Krum falls further to $32.56\%$ at $0\%$ attackers, whereas Mud-HoG reaches $71.23\%$ and FoolsGold reaches $69.87\%$. This establishes a critical baseline: some defenses may appear “stable” under attack simply because they already compress performance in the benign regime, which shifts interpretation from robustness to underfitting.

Second, among defenses that preserve benign accuracy, Multi-Krum and Mud-HoG exhibit the strongest average resilience, but none eliminate worst-case collapses. Aggregating over all nonzero attacker fractions and attacks, MKrum attains the highest mean attacked accuracy on MNIST ($69.83\%$), narrowly ahead of Mud-HoG ($69.05\%$) and FoolsGold ($67.01\%$). On CIFAR-10, Mud-HoG leads at $50.70\%$, with FoolsGold ($49.87\%$) and MKrum ($48.58\%$) close behind; FLGuardian is substantially lower at $39.39\%$. These averages, however, conceal sharp failure modes visible in the heatmaps for both datasets.

\begin{remark}
Even the strongest defenses can suffer severe degradation at specific attacks, implying that the threat cannot be dismissed as “handled” by choosing a robust aggregator alone.
\end{remark}

The sweeps further reveal that defenses interact differently with each attack family. At 20\% attackers on MNIST, MKrum remains comparatively high across all attacks, achieving $83.07\%$ (Grover), $75.59\%$ (Pauli), $79.59\%$ (Bit-flip), and $63.95\%$ (Sign-flip). In contrast, FLGuardian at the same 20\% attackers yields $50.69\%$ (Grover) and $57.01\%$ (Sign-flip), highlighting that screening-based defenses can underperform when adversarial updates retain sufficient alignment with benign update geometry. On CIFAR-10, at 20\% attackers, Mud-HoG achieves the best Grover accuracy ($61.16\%$) and remains competitive under Bit-flip ($55.67\%$).
But, Sign-flip remains challenging across defenses, with Mud-HoG at $77.61\%$ at 20\% attackers yet showing a dramatic vulnerability at 5\% attackers.

Heatmaps, shown in Figures~\ref{fig:mnist-drop-heatmap} and~\ref{fig:cifar-drop-heatmap}, make cross-defense failure modes explicit by exposing whether degradation concentrates in a small subset of aggregators or persists across families.
The concentration of high drop values at low-to-moderate $q$ indicates that the dominant failure mechanism is not purely scaling with attacker prevalence.
Instead, the pattern supports an interpretation in which poisoned updates maintain plausible norms and still improve local loss.

\subsection{Average Variance across Defenses}
Figure~\ref{fig:robustness-index} summarizes defense behavior through a robustness-style comparison between each aggregator’s benign baseline (0\% attackers) and its attacked operating regime (mean across $q>0$). On MNIST, Mud-HoG drops from $91.79\%$ (no attack) to $69.05\%$ (with attack), a $22.74$ pp decrease, and FoolsGold drops from $89.15\%$ to $67.01\%$ ($22.14$ pp). FLGuardian exhibits greater degradation, from $76.53\%$ to $55.22\%$ ($21.31$ pp), while Krum remains low in both regimes (from $48.13\%$ to $46.81\%$), reflecting a defense that primarily reduces capacity rather than improving robustness.

\begin{figure}[h]
    \centering
    \includegraphics[width=1\linewidth]{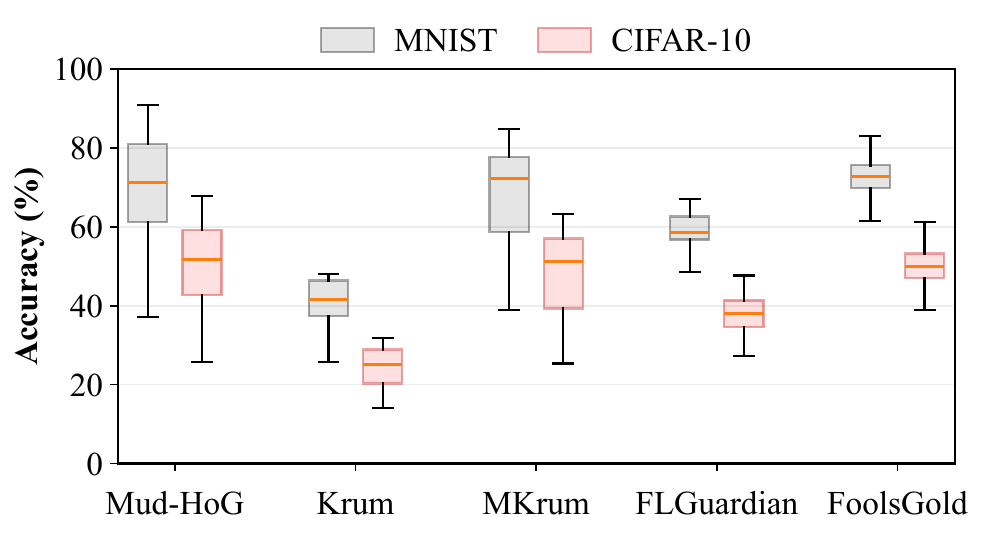}
    \caption{Accuracy variance across the defenses, demonstrating the range of accuracy achieved, summarizing the impact of all attacks.}
    \label{fig:robustness-index}
\end{figure}

On CIFAR-10, the benign-to-attacked gap is similarly large for the high-accuracy defenses: Mud-HoG decreases from $71.23\%$ to $50.70\%$ ($20.53$ pp), and MKrum from $67.82\%$ to $48.58\%$ ($19.24$ pp). FLGuardian drops from $63.31\%$ to $39.39\%$ ($23.92$ \%). These quantified gaps show that robust aggregation improves resilience relative to naive training in many cases, yet persistent quantum attack mechanisms can still impose double-digit average accuracy loss.

\subsection{Why the Attacks Remain Stealthy in Practice}
\label{sec:stealth_discussion}
The empirical results indicate that the proposed attacks can remain difficult to detect under operational monitoring that relies on accuracy trends or simple anomaly thresholds. The following three points support this claim.

First, dilution by aggregation limits the per-round footprint. With 5\% attackers setting, a malicious update contributes a marginal fraction to the aggregated step, and even under robust rules, the adversary can shape the update direction without producing extreme norms. This implies that a persistent attacker can accumulate influence gradually. 

Second, the quantum model’s bounded measurement outputs and the stochasticity induced by non-IID splits increase ambiguity. 
Under Dirichlet heterogeneity ($\alpha=0.9$), benign client updates already differ substantially across rounds, which raises the “background noise floor” that a defender must exceed to reliably detect an attacker. The observed non-monotonic accuracy patterns, in the no-defense case, provide direct evidence that naive heuristics such as “accuracy must degrade monotonically with $q$” are invalid.

Third, the current robust techniques target generic update outliers rather than attack-specific quantum structures. 
The heatmaps show sharp attack-specific pockets of failure, such as worst-case drops above $55$\%. Such pockets are particularly problematic in practice because they can appear as isolated “training instabilities” rather than a persistent adversary, especially when monitoring relies on aggregate accuracy.


\section{Conclusion and Future Scope}
\label{sec:Conclusion and Future Scope}
We proposed a threat model, CULT, that introduces four stealthy circuit-level backdoor attacks targeting both in-training and post-training surfaces of QFL. Along with a strong theoretical analysis, we extensively evaluated (with varying fractions of malicious clients) the proposed attacks on benchmark datasets to assess their impact. 
The results demonstrated the potency of CULT attacks; even with only 5\% malicious client presence, model accuracy degraded by up to 50\% in both with-defense and no-defense cases.

Future work should therefore prioritize defenses that couple update-geometry robustness with {\em quantum-aware} signals. Concretely, integrating circuit-level consistency checks, temporal stability constraints on measurement distributions, and per-client trajectory diagnostics into the aggregation loop may reduce the stealth surface that arises from bounded quantum measurements and heterogeneous client data.

\bibliographystyle{named}
\bibliography{ijcai26}

@inproceedings{McMahan2017,
  title     = {Communication-Efficient Learning of Deep Networks from Decentralized Data},
  author    = {McMahan, H. Brendan and Moore, Eider and Ramage, Daniel and Hampson, Seth and Aguera y Arcas, Blaise},
  booktitle = {Proceedings of the 20th International Conference on Artificial Intelligence and Statistics},
  year      = {2017},
  pages     = {1273--1282}
}

@article{Kairouz2021,
  title     = {Advances and Open Problems in Federated Learning},
  author    = {Kairouz, Peter and McMahan, H. Brendan and Avent, Brendan and Bellet, Aur\'elien and Bennis, Mehdi and Bhagoji, Arjun and Bonawitz, Keith and Charles, Zachary and Cummings, Tyler and Geyer, Robin and others},
  journal   = {Foundations and Trends\textregistered\ in Machine Learning},
  volume    = {14},
  number    = {1--2},
  pages     = {1--210},
  year      = {2021},
  publisher = {Now Publishers}
}

@book{Nielsen2000,
  title     = {Quantum Computation and Quantum Information},
  author    = {Nielsen, Michael A. and Chuang, Isaac L.},
  publisher = {Cambridge University Press},
  year      = {2000},
  isbn      = {978-0521635035}
}

@article{Ren2023,
  title   = {Towards Quantum Federated Learning},
  author  = {Ren, Chao and Yan, Rudai and Zhu, Huihui and Yu, Han and Xu, Minrui and Shen, Yuan and Xu, Yan and Xiao, Ming and Dong, Zhao Yang and Skoglund, Mikael and Niyato, Dusit and Kwek, Leong Chuan},
  journal = {arXiv preprint arXiv:2306.09912},
  year    = {2023}
}

@article{Gurung2023,
  title   = {Quantum Federated Learning: Analysis, Design and Implementation Challenges},
  author  = {Gurung, Dev and Pokhrel, Shiva Raj and Li, Gang},
  journal = {arXiv preprint arXiv:2306.15708},
  year    = {2023}
}

@article{Lu2019,
  title   = {Quantum Adversarial Machine Learning},
  author  = {Lu, Sirui and Duan, Lu-Ming and Deng, Dong-Ling},
  journal = {arXiv preprint arXiv:2001.00030},
  year    = {2019}
}

@inproceedings{Blanchard2017,
  title     = {Machine Learning with Adversaries: Byzantine Tolerant Gradient Descent},
  author    = {Blanchard, Peva and El Mhamdi, El Mahdi and Guerraoui, Rachid and Stainer, Julien},
  booktitle = {Advances in Neural Information Processing Systems 30},
  editor    = {Guyon, I. and Von Luxburg, U. and Bengio, S. and Wallach, H. and Fergus, R. and Vishwanathan, S. and Garnett, R.},
  pages     = {119-129},
  year      = {2017},
  publisher = {Curran Associates, Inc.},
  address   = {Red Hook, NY, USA},
  url       = {https://papers.nips.cc/paper_files/paper/2017/hash/f4b9ec30ad9f68f89b29639786cb62ef-Abstract.html}
}

@inproceedings{Fung2020,
  title     = {Mitigating Sybils in Federated Learning Poisoning},
  author    = {Fung, Clement and Yoon, Chris and Beschastnikh, Ivan},
  booktitle = {Proceedings of the 23rd ACM SIGSAC Conference on Computer and Communications Security},
  year      = {2020},
  pages     = {2946--2961}
}

@inproceedings{Yin2018,
  title     = {Byzantine-Robust Distributed Learning: Towards Optimal Statistical Rates},
  author    = {Yin, Dong and Chen, Ying and Ramchandran, Kannan and Bartlett, Peter L.},
  booktitle = {Proceedings of the 35th International Conference on Machine Learning},
  year      = {2018},
  pages     = {5650--5659}
}

@inproceedings{Bagdasaryan2020,
  title     = {How To Backdoor Federated Learning},
  author    = {Bagdasaryan, Eugene and Shan, Andreas and Veit, Andreas and Hua, Yiqing and Papernot, Nicolas},
  booktitle = {Proceedings of the 23rd International Conference on Artificial Intelligence and Statistics},
  year      = {2020},
  pages     = {2938--2948}
}

@ARTICLE{11003999,
  author={Zhou, Xingjie and Chen, Xianzhang and Liu, Shukan and Fan, Xuehong and Sun, Qiao and Chen, Lin and Qiu, Meikang and Xiang, Tao},
  journal={IEEE Transactions on Information Forensics and Security}, 
  title={FLGuardian: Defending Against Model Poisoning Attacks via Fine-Grained Detection in Federated Learning}, 
  year={2025},
  volume={20},
  number={},
  pages={5396-5410},
  doi={10.1109/TIFS.2025.3570119}}

@inproceedings{gupta2022long,
  title={Long-Short History of Gradients Is All You Need: Detecting Malicious and Unreliable Clients in Federated Learning},
  author={Gupta, Ashish and Luo, Tie and Ngo, Mao V and Das, Sajal K},
  booktitle={European Symposium on Research in Computer Security},
  pages={445--465},
  year={2022},
  organization={Springer}
}

@article{lecun1998gradient,
  title   = {Gradient-Based Learning Applied to Document Recognition},
  author  = {LeCun, Yann and Bottou, L{\'e}on and Bengio, Yoshua and Haffner, Patrick},
  journal = {Proceedings of the IEEE},
  volume  = {86},
  number  = {11},
  pages   = {2278--2324},
  year    = {1998},
  doi     = {10.1109/5.726791}
}

@article{krizhevsky2009learning,
  title       = {Learning Multiple Layers of Features from Tiny Images},
  author      = {Krizhevsky, Alex},
  institution = {University of Toronto},
  year        = {2009},
  month       = apr,
  url         = {https://www.cs.toronto.edu/~kriz/learning-features-2009-TR.pdf}
}

@article{11251187,
  author={Mathur, Aakar and Gupta, Ashish and Das, Sajal K.},
  journal={IEEE Communications Surveys \& Tutorials}, 
  title={When Federated Learning Meets Quantum Computing: Survey and Research Opportunities}, 
  year={2025},
  volume={},
  number={},
  pages={1-1},
  doi={10.1109/COMST.2025.3634143}}

@inproceedings{shen2025label,
  title={Label-free backdoor attacks in vertical federated learning},
  author={Shen, Wei and Huang, Wenke and Wan, Guancheng and Ye, Mang},
  booktitle={Proceedings of the AAAI Conference on Artificial Intelligence},
  volume={39},
  number={19},
  pages={20389--20397},
  year={2025}
}

@inproceedings{ding2025feddlad,
  title={FedDLAD: A Federated Learning Dual-Layer Anomaly Detection Framework for Enhancing Resilience Against Backdoor Attacks},
  author={Ding, Binbin and Yang, Penghui and Huang, Sheng-Jun},
  booktitle={Proceedings of the Thirty-Fourth International Joint Conference on Artificial Intelligence, IJCAI-25},
  pages={5021--5029},
  year={2025}
}

@inproceedings{fung2020limitations,
  title={The limitations of federated learning in sybil settings},
  author={Fung, Clement and Yoon, Chris JM and Beschastnikh, Ivan},
  booktitle={23rd International Symposium on Research in Attacks, Intrusions and Defenses},
  pages={301--316},
  year={2020}
}

@inproceedings{nguyen2022flame,
  title={$\{$FLAME$\}$: Taming backdoors in federated learning},
  author={Nguyen, Thien Duc and Rieger, Phillip and Chen, Huili and Yalame, Hossein and M{\"o}llering, Helen and Fereidooni, Hossein and Marchal, Samuel and Miettinen, Markus and Mirhoseini, Azalia and Zeitouni, Shaza and others},
  booktitle={31st USENIX Security Symposium (USENIX Security 22)},
  pages={1415--1432},
  year={2022}
}

@article{yamany2021oqfl,
  title={OQFL: An optimized quantum-based federated learning framework for defending against adversarial attacks in intelligent transportation systems},
  author={Yamany, Waleed and Moustafa, Nour and Turnbull, Benjamin},
  journal={IEEE Transactions on Intelligent Transportation Systems},
  volume={24},
  number={1},
  pages={893--903},
  year={2021},
  publisher={IEEE}
}

@article{McMahanMRA16,
  author       = {H. Brendan McMahan and
                  Eider Moore and
                  Daniel Ramage and
                  Blaise Ag{\"{u}}era y Arcas},
  title        = {Federated Learning of Deep Networks using Model Averaging},
  journal      = {CoRR},
  volume       = {abs/1602.05629},
  year         = {2016},
  url          = {http://arxiv.org/abs/1602.05629},
  eprinttype    = {arXiv},
  eprint       = {1602.05629},
  bibsource    = {dblp computer science bibliography, https://dblp.org}
}

@article{AshishECAI,
author = {Augello, Andrea and Gupta, Ashish and Lo Re, Giuseppe and Das, Sajal},
year = {2024},
month = {07},
pages = {},
title = {Tackling Selfish Clients in Federated Learning},
doi = {10.48550/arXiv.2407.15402}
}

@misc{neves2025mingling,
  title         = {Mingling with the Good to Backdoor Federated Learning},
  author        = {Neves, Nuno},
  year          = {2025},
  eprint        = {2501.01913},
  archivePrefix = {arXiv},
  primaryClass  = {cs.CR},
  doi           = {10.48550/arXiv.2501.01913},
  url           = {https://arxiv.org/abs/2501.01913}
}

@inproceedings{PaszkeGMLBCKLGA19,
  author       = {Adam Paszke and
                  Sam Gross and
                  Francisco Massa and
                  Adam Lerer and
                  James Bradbury and
                  Gregory Chanan and
                  Trevor Killeen and
                  Zeming Lin and
                  Natalia Gimelshein and
                  Luca Antiga and
                  Alban Desmaison and
                  Andreas K{\"{o}}pf and
                  Edward Z. Yang and
                  Zachary DeVito and
                  Martin Raison and
                  Alykhan Tejani and
                  Sasank Chilamkurthy and
                  Benoit Steiner and
                  Lu Fang and
                  Junjie Bai and
                  Soumith Chintala},
  editor       = {Hanna M. Wallach and
                  Hugo Larochelle and
                  Alina Beygelzimer and
                  Florence d'Alch{\'{e}}{-}Buc and
                  Emily B. Fox and
                  Roman Garnett},
  title        = {PyTorch: An Imperative Style, High-Performance Deep Learning Library},
  booktitle    = {Advances in Neural Information Processing Systems 32: Annual Conference
                  on Neural Information Processing Systems 2019, NeurIPS 2019, December
                  8-14, 2019, Vancouver, BC, Canada},
  pages        = {8024--8035},
  year         = {2019},
  url          = {https://proceedings.neurips.cc/paper/2019/hash/bdbca288fee7f92f2bfa9f7012727740-Abstract.html},
  bibsource    = {dblp computer science bibliography, https://dblp.org}
}

@article{abs181104968,
  author       = {Ville Bergholm and
                  Josh A. Izaac and
                  Maria Schuld and
                  Christian Gogolin and
                  Nathan Killoran},
  title        = {PennyLane: Automatic differentiation of hybrid quantum-classical computations},
  journal      = {CoRR},
  volume       = {abs/1811.04968},
  year         = {2018},
  url          = {http://arxiv.org/abs/1811.04968},
  eprinttype   = {arXiv},
  eprint       = {1811.04968},
  bibsource    = {dblp computer science bibliography, https://dblp.org}
}

@article{Leone2024practicalusefulness,
  doi = {10.22331/q-2024-07-03-1395},
  url = {https://doi.org/10.22331/q-2024-07-03-1395},
  title = {On the practical usefulness of the {H}ardware {E}fficient {A}nsatz},
  author = {Leone, Lorenzo and Oliviero, Salvatore F.E. and Cincio, Lukasz and Cerezo, M.},
  journal = {{Quantum}},
  issn = {2521-327X},
  publisher = {{Verein zur F{\"{o}}rderung des Open Access Publizierens in den Quantenwissenschaften}},
  volume = {8},
  pages = {1395},
  month = jul,
  year = {2024}
}
\end{document}